\documentclass[onecolumn,12pt,draftcls]{IEEEtran}
\usepackage{amsmath}
\usepackage{amssymb}
\usepackage{amsthm}
\usepackage{dsfont}
\usepackage{enumitem}
\usepackage{comment}
\usepackage{bm}
\usepackage{verbatim}
\usepackage{hyperref}
\usepackage{graphicx}

\usepackage{graphicx}
\usepackage{xcolor}

\newcommand{\tilth}{\tilde{\theta}}

\DeclareMathOperator*{\esssup}{ess\,sup}

\newtheorem{theorem}{Theorem}
\newtheorem{prop}{Proposition}
\newtheorem{lemma}{Lemma}
\newtheorem{corollary}{Corollary}
\newtheorem{remark}{Remark}

\newtheorem{assumption}{Assumption}

\newcommand{\bN}{\mathbb{N}}

\newcommand{\bA}{\mathbb{A}}

\newcommand{\cA}{\mathcal{A}}

\newcommand{\cC}{\mathcal{C}}
\newcommand{\cJ}{\mathcal{J}}

\newcommand{\cL}{\mathcal{L}}
\newcommand{\cF}{\mathcal{F}}
\newcommand{\cG}{\mathcal{G}}

\newcommand{\cN}{\mathcal{N}}

\newcommand{\Exp}{{\sf E}}
\newcommand{\Pro}{{\sf P}}

\begin{document}

\title{Quickest Change Detection with Controlled Sensing
}

\author{Venugopal V. Veeravalli, ~\IEEEmembership{Fellow, ~IEEE},
Georgios Fellouris, ~\IEEEmembership{Member,~IEEE}, and  George V. Moustakides, ~\IEEEmembership{Life Senior Member,~IEEE}

\thanks{G. Fellouris is with Department of Statistics, and V.V. Veeravalli is with the ECE Department,
of the University of Illinois at Urbana-Champaign, USA. Email: \{fellouri, vvv\}@illinois.edu. G.V. Moustakides is with the ECE Department, University of Patras, Greece, Email: \ moustaki@upatras.gr.
}

\thanks{This research was supported by the US Army Research Laboratory under Cooperative Agreement W911NF-17-2-0196, and by the US National Science Foundation under grants ATD-1737962 and ECCS-2033900, through the University of Illinois at Urbana-Champaign.}

\thanks{A preliminary version of this work was presented at the 2022 IEEE International Symposium on Information Theory~(ISIT)~\cite{isit_2022_pub}}
}
\IEEEoverridecommandlockouts
\maketitle

\vspace*{-0.5in}

\begin{abstract}
In the problem of quickest change detection, a change occurs at some unknown time in the distribution of a sequence of random vectors  that are monitored in real time, and the goal is to detect this change as quickly as possible subject to a certain false alarm constraint.  In this work we consider this problem  in the presence of parametric uncertainty in the post-change regime and controlled sensing. That is, the post-change distribution contains an unknown parameter, and the distribution of each observation, before and after the change, is affected by a control action. In this context, in addition to a stopping rule that determines the time at which it is  declared that the change has  occurred, one also needs to determine a sequential control policy, which chooses the control action at each time based on the already collected observations. 
We formulate this problem mathematically using Lorden's minimax criterion, and  assuming that there are finitely many possible actions and post-change parameter values. We then propose a specific procedure for this problem that employs an adaptive CuSum statistic in which (i) the estimate of the parameter is based on a fixed number of the more recent observations, and (ii) each action is selected to maximize the Kullback-Leibler divergence of the next observation  based on the current parameter estimate, apart from a small number of exploration times. We show that this procedure, which we call the Windowed Chernoff-CuSum (WCC), is  first-order asymptotically optimal under Lorden's minimax criterion, for every possible possible value of the unknown post-change parameter, as the mean time to false alarm goes to infinity. We also provide simulation results to illustrate the performance of the WCC procedure.

% We establish a universal lower bound on the worst-case detection delay, as the mean time to false alarm goes to infinity, which needs to be satisfied by any procedure for quickest change detection with controlled sensing. 

%We establish a sufficient condition  for the control policy such that when it is coupled with the corresponding CuSum stopping time, the resulting procedure achieves the optimal worst-case detection delay to a first-order asymptotic approximation as the mean time to false alarm goes to infinity, simultaneously for every possible post-change distribution. Finally, we show that this sufficient condition is satisfied when the control maximizes the Kullback-Leibler divergence between the post-change and pre-change density at each time, based on the current estimate of the post-change distribution, apart from a sparse subsequence of time instances where random exploration is performed. 
\end{abstract}
\vspace*{-0.5in}

\section{Introduction} \label{sec:intro}

The problem of detecting changes or anomalies in stochastic systems and time series, often referred to as the quickest change detection (QCD) problem, arises in various branches of science and engineering. The observations of the system are assumed to undergo a change in distribution at the change-point, and the goal is detect this change as soon as possible, subject to false alarm constraints. See \cite{poor-hadj-qcd-book-2009,tart-niki-bass-2014,veer-bane-e-ref-2013,
xie-zou-xie-vvv-jsait-qcd-survey} for books and survey articles on the topic. 

In this paper we study an interesting variant of the QCD problem described in \cite{zhang-mei2020banditQCD_journal,banditQCD}, which the authors call the ``bandit" QCD problem.  In this setting, the distribution of the observations is affected not only by the change but also through a control (action) variable that is chosen by the observer. As described in \cite{banditQCD}, a canonical example application for such a setting is in surveillance systems in which sensors can be switched and steered (controlled) to look for targets in different locations in space, and only a subset of locations can be probed at any given time. The policy for controlling the sensors has to be designed jointly with the change detection algorithm to provide the best tradeoff between detection delay and false alarm rate.   A number of other applications contexts are described in \cite{banditQCD}. 

On a fundamental level, the QCD problem in which the distribution of the observations is affected by control actions falls squarely within the larger context of sequential decision-making problems with observation control (or \emph{controlled sensing} \cite{krishnamurthy2016pomdp}). Such controlled sensing problems have a rich history going back to the seminal work of Chernoff \cite{cher-amstat-1959} on the sequential design of experiments, in which a sequential composite binary hypothesis testing problem with observation control is studied. Other works on sequential hypothesis testing with observation control include \cite{Bessler1960_I, Bessler1960_II,albert1961,Kiefer_Sacks_1963,Lalley_Lorden_1986,Keener_1984} and  more recently \cite{niti-atia-veer-ieeetac-2013,naghshvar2013active,niti-veer-sqa-2015,deshmukh2021sequential}. There has also been considerable progress on the special ``multi-channel" case of sequential hypothesis testing with observation control, which is also commonly referred to as sequential anomaly detection.  In this context, there are multiple data streams,  some of which are anomalous, and the goal is to accurately pick out the anomalous ones among them, while observing only a subset of the streams at each time-step \cite{Cohen2015active,Cohen2019nonlinearcost,Cohen2020composite,oddball_2018,Tsopela_2019,Tsopela_2020,tsopelakos2021sequential}. The QCD problem in the multi-channel setting with observation control has been studied  in \cite{Draga96}, and more 
recently in \cite{zhang-mei2020banditQCD_journal,Mei_Moust_2021,Anamitra2021,XuMeiPostUncert_SeqAn}. In this context,  an unknown subset of the streams undergo a change in their  distributions at the unknown change-point, and only a subset of them  can  be observed at each time-step. 
%In \cite{zhang-mei2020banditQCD},  a control (sampling) strategy based on the Thompson sampling approach from the multiarmed bandit literature \cite{lattimore2020bandit} is proposed and combined with the Shiryaev-Roberts QCD rule \cite{pollak1987}. %to develop a procedure for detecting the change.  % While they provide only a limited theoretical analysis of their procedure, they show through extensive simulations that their procedure is effective.
%In  \cite{Mei_Moust_2021} it is assumed that the change occurs in only data stream and that it is possible to sample only one at a time. 
%In this context, it is proposed  to run the Cumulative Sum (CuSum) statistic \cite{page1954} of a data stream until it either  either regenerates, in which case the next data stream is sampled, or it exceeds a threshold, in which case the alarm is raised. It is shown that the incurred performance loss of this algorithm, as measured by Lorden's \cite{lorden1971} criterion, relative to the performance of the oracle rule that knows a priori the identity of the data stream, remains bounded as  the false alarm rate goes to 0, a property known as second-order asymptotic optimality. 
%A similar algorithm is applied in \cite{Anamitra2021} in the case that a sequence of multivariate Gaussian observations is monitored and there is a change in their covariance matrix. In this context, a first-order asymptotic optimality is established under the assumption that it is possible to sample only two datastreams at a time, but  more than two may change. 

Our work is  inspired by \cite{banditQCD}, in which a general setting for QCD with controlled sensing is considered. In particular, it is assumed that the pre- and post-change distributions can be affected by the choice of a control action, which takes values in a finite set. Furthermore, the post-change is allowed to have parametric uncertainty within a finite parameter set.

The formulation of the optimization problem to obtain the best tradeoff between detection delay and false alarm rate in \cite{banditQCD}  is different from the standard formulations of the QCD problem \cite{veer-bane-e-ref-2013}. In particular, the false alarm constraint used in  \cite{banditQCD} is one where the probability of stopping (and declaring a change) before a \emph{fixed time} $m$ under the pre-change regime is constrained to some level $\alpha$. Under such a constraint, if the change  happens at any time after $m+1$, then the delay can be made to be equal 0 by simply stopping at time $m+1$, if the test did not stop before, while maintaining the false alarm constraint. The authors of \cite{banditQCD} refer to \cite{lai1998} to justify their false alarm constraint, but the false alarm constraint proposed in \cite{lai1998} requires the probability of stopping in \textit{any} time-window of size $m$ in the pre-change regime be small (see \cite[Section II.B]{lai1998}), not just the window from $0$ to $m$. Furthermore, it is more common in the QCD literature to constrain the \emph{mean time to false alarm} 
%(or its reciprocal) 
\cite{veer-bane-e-ref-2013}, which is also a constraint used in \cite[Section I]{lai1998}.

%(\textcolor{red}{Should we say somehow that it is a weaker/more general constraint? Let's leave this for the journal version}).

%Another issue with the formulation in \cite{banditQCD} is that the measure of delay being used is not specified clearly. 

The measure of delay used in \cite{banditQCD} is the expected delay conditioned on a fixed change-point. In QCD problems where no prior is assumed on the change-point, it is common to measure delay by  taking the supremum of the expected delay over all possible change-points (see, e.g., the Lorden and Pollak formulations \cite{veer-bane-e-ref-2013}). Furthermore, the
asymptotic upper bound on the delay for a fixed change-point) given in \cite[Theorem 3]{banditQCD} is larger than the corresponding lower bound given in \cite[Theorem 2]{banditQCD} (which assumes that the change happens at time 1) by a factor greater than 160. In contrast, first-order asymptotic optimality results in the QCD literature require the ratio of the upper and lower bounds on the delay metric to converge to 1, as the false alarm rate goes to 0 \cite{veer-bane-e-ref-2013}.

The  $\epsilon$-GCD procedure proposed in \cite{banditQCD} for change detection with observation control uses at each time instant  a maximum likelihood estimate (MLE) of the post-change parameter to determine the best action at each time step, except that with a fixed probability  $\epsilon$, the action is chosen uniformly at random. The MLE  at each time instant is determined only by those samples that resulted from random exploration. The use of the current maximum likelihood estimate for determining the current action is in fact the key feature in Chernoff's proposed control policy in  \cite{cher-amstat-1959} for the problem of sequential composite binary hypothesis testing problem with observation control. The need for random exploration of actions in that context arises because the true post-change parameter may not be distinguishable from other possible post-change parameter values under certain actions. However, exploring at random with probability  $\epsilon>0$ \textit{at each time instant} will generally lead to substantial performance loss (by roughly a factor of $1/(1-\epsilon)$) relative to that of an oracle that knows the post-change parameter.   In contrast, as we show in this paper,  we can achieve the  performance of the latter to a first-order asymptotic approximation as the false alarm rate goes to 0, by performing random exploration at only $o(n)$ time-steps over a time-horizon of length $n$.

Another consideration with the use of the MLE that relies on data from the beginning  to estimate the post-change parameter in the QCD setting (in contrast to the sequential hypothesis testing setting) is that it could  potentially be biased away from the true post-change parameter, due to the pre-change observations, if  the change-point is not small.
This means that it will take longer for the MLE to converge to the true post-change parameter in the post-change regime, thereby increasing the delay, as the change-point gets larger. We get around this problem by forcing the MLE to use only a window of past observations, and by using a windowed CuSum test (as is done in \cite{xie-mous-xie-IT-2023} for the problem of quickest change detection with post-change parametric uncertainty, without controlled sensing).

%

%Therefore, it is not  in general possible for the $\epsilon$-GCD procedure to be asymptotically optimum.
%and such an exploration strategy does not suffer from a loss in asymptotic optimality.

To sum up, our goal in this paper to precisely formulate the QCD problem with controlled sensing, as well as to propose and analyze a novel,
asymptotically optimal algorithm for this problem. As described in  Section~\ref{sec:prob_form}, we use Lorden's metric \cite{lorden1971} for the delay, and we pose the optimization problem as the minimization of this delay metric, under a constraint on the mean time to false alarm (MTFA). In Section~\ref{sec:ULB}, we derive a universal lower bound on the delay of any procedure, under the MTFA constraint. In Section~\ref{sec:Chernoff_CuSum}, motivated by the discussion in the previous paragraphs, we develop a procedure, which we call the Windowed Chernoff-CuSum (WCC) procedure. We establish that the WCC procedure is asymptotically optimal under Lorden's criterion as the MTFA goes to infinity.  In Section~\ref{sec:sim}, we provide some simulation results that illustrate the main points of the theoretical analysis.

\section{Problem formulation}\label{sec:prob_form}

Let $\{X_n: n \in \bN\}$ be a sequence of random vectors  whose values are  observed  sequentially,  let $\{U_n: n \in \bN \}$ be a sequence of random variables to be used for randomization purposes, and let  $\{\cF_n,: n \in \bN\}$ the filtration generated by these two sequences, i.e.,
%$$\cF_1 =\sigma(X_1,U_1), \quad \cF_n =\sigma( \cF_{n-1}, X_n, U_n), \quad n > 1.$$
\begin{equation} \label{eq:filtration}
\cF_n := \sigma( X_1, \ldots, X_n, U_1, \ldots, U_n), \quad n \in \bN.
\end{equation}
We also denote by $\cF_0$ the trivial $\sigma$-algebra and for any $m,n \in \bN$ with $m \leq n$ we set:
\begin{equation} \label{eq:filtration2}
\cF_{m:n} := \sigma( X_{m}, \ldots, X_n, U_m, \ldots, U_n).
\end{equation}
\begin{assumption} \label{ass:basic}
We assume that, for any $n \in \bN$,  $U_n$ is independent of $\cF_{n-1}$ and   uniformly distributed in $[0,1]$,  and that  $X_n$  is independent of $U_n$ and conditionally independent of $\cF_{n-1}$ given the value of a control/action  $A_{n}$. We assume that the  action  $A_{n}$ takes values in a  finite  set  $\bA$, with at least two elements. Furthermore,  for $n > 1$, $A_n$ is a $\cF_{n-1}$-measurable random variable, with $A_1$ being uniformly distributed on the set $\bA$.
\end{assumption}

We refer to the sequence of actions $A := \{A_n: n \in \bN \}$ as  a \emph{control policy} and we denote by $\cA$ the \emph{family} of all control policies, i.e.,   $A= \{A_n: n \in \bN \} \in \cA$.

Let $\{f^{\theta}_a: \theta \in \Theta, a \in \bA\}$ be a set of densities with respect to a dominating measure,  $\lambda$, where  $\Theta$ is an arbitrary finite set. Furthermore, for some $\theta_0 \notin \Theta$, let $\{f^{\theta_0}_a: a \in \bA\}$ be a set of densities with respect to $\lambda$. 

Let the \textit{change-point} be denoted by $\nu \in \bN$, which we assume is completely unknown and deterministic. Given that  $A_{n}=a$, where $a \in \bA$,  we assume that for $n < \nu $ (pre-change), $X_n$ has conditional density $f^{\theta_0}_{a}$, and that for  $n \geq \nu$ (post-change), $X_n$ has conditional density $f^{\theta}_{a}$, with $\theta \in \Theta$.
%\textcolor{red}{and $\theta$ being unknown \emph{a priori}}.

To be more specific,  we denote by $\Pro_{\nu,A}^{\theta}$ 
the underlying probability measure, and by $\Exp_{\nu,A}^{\theta}$ the corresponding expectation, when the change-point is $\nu$,  the post-change parameter  $\theta \in \Theta$, and  the control policy $A$ is used, which means that for any 
$n \in \bN$ and any Borel set $B$  we have
$$
\Pro_{\nu,A}^\theta(X_n \in B \,|\, \cF_{n-1})= 
\begin{cases}
\int_B f^{\theta_0}_{A_n} \; d\lambda, \quad \text{if} \quad n < \nu\\
\int_B f^{\theta}_{A_n} \; d\lambda, \quad \text{if} \quad n \geq \nu.
\end{cases}
$$

Moreover, we denote  by
$\Pro_{\infty,A}$ the underlying probability measure, and by $\Exp_{\infty,A}$  the corresponding expectation, when the change never occurs and the control policy $A$ is used, which means that for any $n \in \bN$ and any Borel set $B$  we have
$$\Pro_{\infty,A}(X_n \in B \,|\, \cF_{n-1}) =\int_{B} f_{A_{n}}^{\theta_0} \; d\lambda. 
$$

%Then, for any $\theta \in \Theta$, $a \in \bA$, $\nu \in \bN \cup \{0\}$,  $n \in \bN$ and any Borel set $B$  we have 
%$$\Pro^{\theta}_{\nu}(X_n \in B \;  | \; \cF_{n-1}, A_{n}=a) =
%\begin{cases}
%\int_B f_{a} \; d\lambda, \quad \text{if} \quad n < \nu\\
%\int_B g^{\theta}_{a} \; d\lambda \quad \text{if} %\quad n \geq \nu,
%\end{cases}
%$$

 A \emph{procedure} for quickest change detection with controlled sensing consists of a pair $(A,T)$, where $A$ is a control policy, i.e.,  $A \in \cA$, and 
$T$ is an $\{\cF_n\}$-stopping time, i.e., $\{T=n\} \in \cF_n$ for every $n \in \bN$. We denote by $\cC$ the family of all  procedures, i.e., $(A,T)\in \cC$.  \\

%\sout{Moreover, we denote by  $\Pro_{\infty}$ the underlying probability measure, and by $\Exp_{\infty}$  the corresponding expectation, when the change never occurs. }

\noindent \textbf{False Alarm Measure.} We measure the false alarm performance of a procedure in terms of its mean time to false alarm, and we denote by $\cC_\gamma$  the subfamily of procedures for which the mean time to false alarm is at least $\gamma$, i.e., 
\[
\cC_\gamma= \left\{ (A,T) \in \cC: \; \Exp_{\infty,A}[T] \geq \gamma \right\}.
\]

\noindent \textbf{Delay Measure.} We use a worst-case measure for delay, that is  the commonly used Lorden's measure \cite{lorden1971}. Specifically, for any $\theta \in \Theta$ and $(A,T) \in \cC$ we set 
\begin{align} 
 \cJ_\theta(A,T) & :=  \; \sup_{\nu \geq 1} \; \text{ess sup} \; \Exp_{\nu,A}^{\theta} \left[(T-\nu+1)^+  |  \; \cF_{\nu-1} \right]  
 %& =  \; \sup_{\nu \geq 1} \; \text{ess sup} \; \Exp_{\nu,A}^{\theta} \left[T-\nu+1  |  \; \cF_{\nu-1} \textcolor{blue}{; T \geq \nu } \right].
 \label{eq:Lorden_meas}
\end{align}

\noindent \textbf{Optimization Problem.} The optimization problem we consider  is to find a test that can be designed to belong to $\cC_\gamma$ for every $\gamma>1$\textcolor{green}{,} and achieves 
\begin{equation}\label{eq:opt_prob}
\inf_{ (A,T) \in \cC_\gamma}  \cJ_\theta(A,T) 
\end{equation}
 to a first-order asymptotic approximation as 
$ \gamma \to \infty$ simultaneously for every post-change parameter $\theta \in \Theta$.

\noindent We make the following further assumptions in our analysis:

\begin{assumption}\label{ass:a2}
For every $\theta \in \Theta$ there exists an $a \in \bA$ such that the Kullback-Leibler (KL) divergence between the densities $f_a^\theta$ and $f_a^{\theta_0}$ is positive, i.e.,
\begin{equation} \label{eq:Iathetadef}
I_a^{\theta} := D(f_a^\theta\| f_a^{\theta_0}) =  \int \log (f_a^\theta/ f_a^{\theta_0}) \; f_a^\theta \; d \lambda >0.
\end{equation}
This means that the post-change distribution is distinguishable from the pre-change distribution for at least one choice of control. 
\end{assumption}
If Assumption~\ref{ass:a2} does not hold for some $\theta \in \Theta$, then the change will not be detected efficiently if the true post-change parameter is $\theta$. This assumption implies that 
\begin{equation}
I^{\theta} :=  \max_{a \in \bA}   I_a^{\theta}  >0, ~ \forall \, \theta \in \Theta. 
\end{equation}
\begin{assumption}\label{ass:a1}
For every $\theta \in \Theta$ and 
$a \in \bA$ we have 
\begin{equation}
\int \log (f_a^\theta/ f_a^{\theta_0})^2 \; f_a^\theta \; d \lambda < V <\infty.   
\end{equation}
\end{assumption}

Assumption~\ref{ass:a1} is a technical condition that is needed in our analysis (see, e.g.,  Theorem~\ref{th:LB} and Theorem \ref{th:aub}).

\begin{assumption}\label{ass:a3}
For every $\theta, \tilth \in \Theta$, such that $\theta \neq \tilth$, there exists an $a \in \bA$ so that:
\begin{equation} \label{eq:cons_assumption}
D(f_a^\theta\| f_a^{\tilth}) = \int \log (f_a^\theta/ f_a^{\tilth}) \; f_a^\theta \; d \lambda >0,
\end{equation}
i.e., the post-change distribution for two distinct values of the post-change parameter are distinguishable by at least one control. 
\end{assumption}
Assumption~\ref{ass:a3} is crucial for the consistent estimation of $\theta$ in the post-change regime (see Lemma~\ref{lem:cons}).  

\begin{remark}\label{rem:a3} One could strengthen Assumption~\ref{ass:a3} by assuming that 
\begin{equation*} 
D(f_a^\theta\| f_a^{\tilth})  >0,
\end{equation*}
for all $\theta, \tilth \in \Theta$, such that $\theta \neq \tilth$, and for all $a\in \bA$. However, this stronger assumption fails to hold in many applications of interest. 

For example, consider a multichannel setting with $K$ streams (see, also, Section~\ref{sec:sim}) in which a single (unknown) stream of observations undergoes a change in distribution at the change-point, and only one of the streams is observed at each time step. Here $\Theta= \bA = \{1,2, \ldots, K\}$. If $\theta = j$, then the $j$-th stream undergoes a change in distribution, and if $a=i$, the $i$-th channel is observed. In particular, if $a=1$, then for $\theta=2$ and $\tilth=3$, the observed channel (channel 1) has the pre-change distribution, which means that:
\begin{equation*} 
D(f_1^2\| f_1^3) =0.
\end{equation*}
Thus, this strengthened version of 
Assumption~\ref{ass:a3} fails to hold in this case.
\end{remark}

For any $\theta, \tilde{\theta} \in \Theta$, such that  $\tilth \neq \theta$,
and $a \in \bA$, we further define the Bhattacharya coefficient (see, e.g. \cite{moulin-veeravalli-2018}):
\begin{equation} \label{bhat}
    \rho (\theta, \tilth, a) := \int \sqrt{f_a^\theta\;  f_a^{\tilth}} \: d\lambda.
\end{equation}
By Assumptions~\ref{ass:a2} and \ref{ass:a3}, it follows   that
for any $\theta, \tilde{\theta} \in \Theta$, such that  $\tilth \neq \theta$, there exists an $a\in \bA$ such that $\rho (\theta, \tilth, a) < 1$, and as a result
\begin{equation} \label{bhat_average}
\rho (\theta, \tilth) := \frac{1}{|\bA|} \sum_{a\in \bA} \rho (\theta, \tilth, a) < 1.
\end{equation}
We also denote by $\rho$ the maximum of these quantities:
\begin{align}\label{max_Bhat}
\rho &:= \max_{\tilth \neq \theta} \rho (\theta, \tilth).
\end{align}
The quantity $\rho$, which is less than 1, will play a key role in the consistent estimation of $\theta$ in the post-change regime (see Lemma~\ref{lem:cons}), as well as in lower bounding the post-change expected values of the increments of the proposed detection statistic (see    Lemma~\ref{lemma:LB on Exp}). 

\smallskip

\noindent{\textbf{Notation:}} The notation $f(x) \sim g(x)$ as $x\to x_0$ means that $\lim_{x \to x_0} \frac{f(x)}{g(x)} = 1$. The minimum of two numbers $x$ and $y$ is denoted as $x\wedge y$. For $k > \ell$, the set of variables $\{x_\ell, x_{\ell+1}, \ldots, x_k\}$ is sometimes denoted compactly as $x_{[\ell,k]}$.
\section{Universal lower bound} \label{sec:ULB}
For $\theta \in \Theta$, $A \in \cA$, we define the log-likelihood ratio of the observation at time $m \in \bN$ as
\begin{align} \label{Lambda}
\Lambda^\theta_{m, A}  :=  \log \left( \frac{f^\theta_{A_{m}}(X_m)}{f^{\theta_0}_{A_{m}}(X_m)} \right),
\end{align}
 and we observe that for any  $\nu, t \in \bN$, we have 
\begin{align*}
 \frac{d\Pro^\theta_{\nu,A}}{d\Pro_{\infty,A}} (\cF_{\nu+t})= \exp \left\{  \sum_{m=\nu}^{\nu+t}  \Lambda^{\theta}_{m,A} \right\}.
 \end{align*}
In the following result, we provide an asymptotic lower bound on the value of the optimization problem in \eqref{eq:opt_prob} as the MTFA goes to infinity, i.e., as $\gamma \to \infty$. 
\begin{theorem}\label{th:LB}
For any $\theta \in \Theta$, as $\gamma \to \infty$ we have
$$\inf_{ (A,T) \in \cC_\gamma}  \cJ_\theta(A,T) \geq \frac{\log \gamma}{ I^\theta} (1+o(1)).$$
\end{theorem}

\begin{proof}
Fix $(A,T) \in \cC_\gamma$ and $\theta \in \Theta$. By \cite[Theorem  1]{lai1998} it suffices to show that, for every $\delta>0$, the sequence 
\[
\sup_{\nu \in \bN} \esssup \Pro^\theta_{\nu,A} 
\left( \max_{1 \leq t \leq n} \sum_{m=\nu}^{\nu+t} \Lambda^{\theta}_{m,A} > I^\theta (1+\delta) n \; |\;  \cF_{\nu-1}  \right) 
\]
converges to 0 as $n \to \infty$.  Indeed, for every  $n, \nu \in \bN$ we have 
\begin{align*}
&\Pro^\theta_{\nu,A}  \left( \max_{1 \leq t \leq n} \sum_{m=\nu}^{\nu+t} \Lambda^{\theta}_{m,A}  > I^\theta (1+\delta) n \; |\;  \cF_{\nu-1}  \right) \\
&\leq \Pro^\theta_{\nu,A}  \left( \max_{1 \leq t \leq n} \sum_{m=\nu}^{\nu+t} (\Lambda^{\theta}_{m,A} - I^{\theta}_{A_{m}})  > I^\theta \delta  n \; |\;  \cF_{\nu-1}  \right) \\
&\leq  \frac{n V}{(I^\theta \delta  n)^2}=  \frac{ V}{(I^\theta \delta)^2 n}.
\end{align*}

Here,  the first inequality follows from the fact that $I_a^\theta \leq I^\theta$ for every $a \in \bA$, and the second one from a conditional version of  Doob's submartingale inequality. We can apply the latter because 
$$\left\{Y^\theta_{\nu:t,A} := \sum_{m=\nu}^{\nu+t} (\Lambda^{\theta}_{m,A} - I^{\theta}_{A_{m}}), \cF_{\nu+t}, t \in \bN \right\}$$
is a $\Pro^\theta_{\nu,A}$-martingale and by assumption \textbf{A1} it follows that 
$$\text{Var}^{\theta}_{\nu,A}[Y^\theta_{\nu:t, A} \, | \, \cF_{\nu-1}] \leq V \, t \quad \text{for all} \quad t \in \bN.$$ 
\end{proof}
In the next section  we propose a specific procedure for quickest change detection with controlled sensing, and establish that it achieves the lower bound in Theorem~\ref{th:LB} asymptotically as $\gamma \to \infty$.

\section{The Windowed Chernoff-CuSum Procedure} \label{sec:Chernoff_CuSum}
%
%We now specialize the CuSum procedure $(A,T_b)$ introduced in the previous section by fixing the control policy. 

Consider any control policy $A \in \cA$, and window of size $w \in \bN$. At time $n > w$, the maximum likelihood estimator (MLE) of $\theta$ based on observations $\{X_m: m=n-w, \ldots, n-1\}$ is given by: 
\begin{align} \label{eq:ML}
\hat{\theta}_{n} & \in \; \arg \max_{\vartheta \in \Theta} \sum_{m=n-w}^{n-1}  \log f^{\vartheta}_{A_m} (X_m).
\end{align}

To define the proposed policy, we also need to introduce a  sequence of  deterministic times, 
$$\cN \subseteq\{w+1, w+2, \ldots\},$$  
which contains at most $q$ elements 
in any interval of length $w$, i.e., 
\begin{align} \label{sampling}
\sum_{i=w+1}^{w+n}  \mathbb{I}\{i \in \cN\}\leq  \frac{n}{w}  q, \quad \forall \; n \in \bN, 
 \end{align}
for some positive integer $q$ that is smaller than $w$. The sequence of times $\cN$ will be used for random exploration of controls, which will be crucial for consistent estimation of $\theta$ in the post-change regime, under Assumption~\ref{ass:a3}.

 \begin{comment}
\textcolor{red}{
\[
\cN_\eta := \bigcup_{m=1}^\infty \{mw + \lceil \eta^\ell \rceil: \ell = 1, \ldots, q\}, \quad \text{for some} \; \eta>1.
\]
where  $q=\lfloor\log w/\log \eta \rfloor$. That is,  in each window of size $w$, there are $q=\lfloor\log w/\log \eta \rfloor$ time instances which are used for exploration.
Note:  In order not to suffer from suboptimality in the first-order, we could use any sublinear sampling scheme for exploration, e.g., $\sqrt{w}$.}
\end{comment}

%\[
%\{1\} \cup \{\lceil \eta^\ell \rceil:  \ell \in \bN \}, %\quad \text{for some} \; \eta>1.
%\]
%\textcolor{red}{This sequence needs to be modified so that in any interval of length $w$ we sample the control  randomly at some sublinear (e.g., $\log w$ number of times) }

Given such a sequence, we propose a control policy $A^*$ as follows.
\begin{itemize}
    \item For $n=1, \ldots, w$, $A_n^*$ is selected uniformly at random on the set $\bA$, using  the randomization variable $U_{n-1}$.
    \item If $n > w$, and $n \in \cN$, then   $A_n^*$ is selected uniformly at random on $\bA$, using  the randomization variable $U_{n-1}$. 
    \item If $n > w$, and $n \notin \cN$, 
    $A_n^*$ is selected to maximize the Kullback-Leibler divergence of the post-change versus the pre-change distribution based on $\hat{\theta}_{n}$, i.e., 
\begin{equation}\label{eq:Anstar}
A_n^* \in \arg \max_{a \in \bA} I_a^{\hat{\theta}_{n}} = 
\arg \max_{a \in \bA} D(f_a^{\hat{\theta}_{n}}\| f_a^{\theta_0})
.
\end{equation}
\end{itemize}
The proposed stopping time is defined using a windowed CuSum test, which was introduced in \cite{xie-mous-xie-IT-2023}. Specifically,  for any $A \in \cA$, and for $n > w$, we define the following CuSum-like statistic recursively:
\begin{equation} \label{eq:CuSum_stat}
W_{n,A}= \max\{ W_{n-1,A}, 0\} + \Lambda^{\hat{\theta}_{n}}_{n,A}, \quad  \quad W_{w,A}=0,
\end{equation}
where the log-likelihood ratio $\Lambda$ is as defined in \eqref{Lambda}. 
The change is declared at the first  time this statistic exceeds a threshold $b$, i.e.,
\begin{equation} \label{eq:Tba}
T_{b,A}:= \inf\left\{n > w:  W_{n,A}  \geq b\right\}.
\end{equation}
where $b$ is to be selected to satisfy the false alarm constraint. 

We refer to the pair $(A^*, T_{b, A^*})$ as the  \emph{ Windowed Chernoff-CuSum} (WCC) procedure to acknowledge that Chernoff \cite{cher-amstat-1959} was the first to suggest a control policy similar to $A^*$ for the problem of sequential hypothesis testing with controlled sensing.  

In the remainder of this section, we  establish  bounds on the performance of the WCC procedure. We begin by establishing a choice of the threshold $b$ that guarantees that the MTFA constraint is met.

\begin{lemma} \label{lem1}
For any  $A \in \cA$ and  $\gamma>1$, by setting $b=\log \gamma$,  we can ensure that
 $(A, T_{b,A}) \in \cC_\gamma$, i.e., that $\Exp_{\infty,A}[T_{b,A}] \geq \gamma.$
\end{lemma}

\begin{proof}
We can rewrite the recursion for  test statistic in \eqref{eq:CuSum_stat} as:
\begin{equation} \label{eq:CuSum_stat2}
e^{W_{n,A}}= \max\left\{ e^{W_{n-1,A}}, 1\right\} \frac{f^{\hat{\theta}_{n}}_{A_{n}}(X_n)}{f^{\theta_0}_{A_{n}}(X_n)},~n> w,  \quad e^{W_{w,A}}=1.
\end{equation}
We now define a new statistic $\{R_{n,A}\}$ akin to the Shiryaev-Roberts statistic (see, e.g., \cite{veer-bane-e-ref-2013}), which has the recursion:
\begin{equation} \label{eq:SR_stat2}
R_{n,A} := (R_{n-1,A} + 1) \, \frac{f^{\hat{\theta}_{n}}_{A_{n}}(X_n)} {f^{\theta_0}_{A_{n}}(X_n)},~n> w,  \quad R_{w,A}:=1.
\end{equation}
Clearly, $R_{n,A} \geq e^{W_{n,A}}$, for all $n\geq w$.
%, and therefore for any $b>0$ we have 
%\[
%T_{b,A} \leq T'_{b,A} :=  \inf\{n > w: R_{n,A} \geq e^{b}\}.
%\]

It is easily seen that $\{R_{n,A} - n: n \geq w\}$ is a martingale sequence under $\Pro_{\infty,A}$, with respect to the filtration $\{\cF_n: n \in \bN\}$. Indeed, since  $A_n$ is $\cF_{n-1}$-measurable by Assumption~\ref{ass:basic}, and $\hat{\theta}_n$ is $\cF_{n-1}$-measurable by construction (see \eqref{eq:ML}), we have for $n>w$:
\[
\Exp_{\infty,A} \left[ R_{n,A} - n \; |\;  \cF_{n-1} \right] = R_{n-1,A} + 1 -n = R_{n-1,A} - (n-1).
\]
Then, by applying the Optional Sampling Theorem \cite{tartakovsky_book_2014}:
\[
\Exp_{\infty,A} \left[ R_{T_{b,A},A} - T_{b,A}\right] = 
R_{w,A} -w = 1-w,
\]
which implies that
\[
\Exp_{\infty,A} \left[T_{b,A}\right] = \Exp_{\infty,A} \left[ R_{T_{b,A},A}\right] + w-1 \geq \Exp_{\infty,A} \left[ e^{W_{T_{b,A},A}}\right] \geq e^b.
\]
Therefore, setting $b=\log \gamma$ ensures that
 $(A, T_{b,A}) \in \cC_\gamma$.
\end{proof}

Note that Lemma~\ref{lem1} holds for any control policy $A\in \cA$ satisfying Assumption~\ref{ass:basic}, and not just the control policy $A^*$.

Our main goal in the remainder of this section is to analyze the delay of   the WCC procedure as $\gamma \to \infty$. Towards this end, we first establish an upper bound on our measure of delay for any control policy $A\in \cA$ (see \eqref{eq:Lorden_meas}):
\[ 
\cJ_\theta(A,T_{b,A}) = 
\sup_{\nu \geq 1} \; \text{ess sup} \; \Exp_{\nu,A}^{\theta} \left[(T-\nu+1)^+  |  \; \cF_{\nu-1} \right].
\]

\begin{lemma} \label{lem:worst-case-bd}
For any change-point $\nu \in \bN$,  control policy $A\in \cA$, and threshold $b>0$
\[ \cJ_\theta(A,T_{b,A})
%\mathrm{ess~sup} \; \Exp_{\nu,A}^{\theta} \left[(T_{b,A}-\nu+1)^+  |  \; \cF_{\nu-1} \right] 
\leq     \sup_{a_1, \ldots, a_{w}} \Exp_{1,A}^{\theta} \left[T_{b,A}   \, |   \,  A_{1}=a_1, \ldots, A_{w}=a_w  \right].
\]
\end{lemma}
\begin{proof}
\begin{comment}
Suppose first that $\nu=1$. Then it is clear that the left hand side (LHS) in the lemma equals $\Exp_{1,A}^{\theta} \left[T_{b,A}\right]$ and 
\begin{align*}
\Exp_{1,A}^{\theta} \left[T_{b,A}\right] 
& =  \Exp_{1,A}^{\theta} \left[\Exp_{1,A}^{\theta} \left[T_{b,A} | \;\cF_w \right]\right]\\
& \leq  \max_{a_1,\ldots, a_{w+1}} \Exp_{1,A}^{\theta} \left[T_{b,A}| \; A_1=a_1, \ldots, A_{w+1}=a_{w+1}\right].
\end{align*}
\end{comment}

For each $\nu \geq  1$, we denote by $T_{b,A}^\nu$ a version of the stopping time in \eqref{eq:Tba} that 
%\textit{\textcolor{red}{ignores all observations up to time $\nu-1$ and} 
starts computing the test statistic from time $\nu+w$, with initialization  0 at time $\nu+w-1$, i.e.,
\begin{equation} \label{eq:Tbam}
T_{b,A}^\nu:= \inf\left\{ n >\nu+w-1 :  W_{n,A}^\nu  \geq b\right\}.
\end{equation}

with
\begin{align} \label{eq:CuSum_stat_m}
\begin{split}
W_{n,A}^\nu &= \max\{ W_{n-1,A}^\nu, 0\} + \Lambda^{\hat{\theta}_{n}}_{n,A},~~ n > \nu+w-1, \\
%~\text{with}~~
W_{\nu+w-1,A}^\nu &=0.
\end{split}
\end{align}
Then, for any  $\nu \geq  1$ we have 
\begin{align*}
\Exp_{\nu,A}^{\theta} \left[(T_{b,A}-\nu+1)^+ \; |  \; \cF_{\nu-1} \right] & \stackrel{\text{(a)}}{\leq} w + \Exp_{\nu,A}^{\theta} \left[(T_{b,A}-\nu+1-w)^+ \; |  \; \cF_{\nu-1} \right]\\
&\stackrel{\text{(b)}}{\leq}  w+\Exp_{\nu,A}^{\theta} \left[T^\nu_{b,A}-\nu+1 -w \; |  \; \cF_{\nu-1} \right]\\
& \stackrel{\text{(c)}}{=}  \Exp_{\nu,A}^{\theta} \left[\Exp_{\nu,A}^{\theta} \left[T_{b,A}^{\nu} -\nu+1  \;  |   \; A_{[\nu,\nu+w-1]} , \; \cF_{\nu-1}\right]\; | \; \cF_{\nu-1} \right]\\
& \stackrel{\text{(d)}}{=}  \Exp_{\nu,A}^{\theta} \left[ \Exp_{\nu,A}^{\theta} \left[T_{b,A}^{\nu} -\nu+1   \; |   \; A_{[\nu,\nu+w-1]} \right]
  \; | \;  \cF_{\nu-1}  \right]\\
%& \textcolor{red}{=  \Exp_{\nu,A}^{\theta} \left[T_{b,A}^{\nu}-\nu+1 \; |   \; A_{[\nu+1,\nu+w]} \right]} \\
& \leq     \sup_{a_1, \ldots, a_{w}} \Exp_{\nu,A}^{\theta} \left[T_{b,A}^{\nu} -\nu+1  \;   |   \;  A_{\nu}=a_1, \ldots, A_{\nu+w-1}=a_w  \right]\\
& \stackrel{\text{(e)}}{=}   \sup_{a_1, \ldots, a_{w}} \Exp_{1,A}^{\theta} \left[T_{b,A} \; |   \;  A_{1}=a_1, \ldots, A_{w}=a_w  \right],
\end{align*}
where step (a) holds because $x^+\leq w+ (x-w)^+$ for any $x, w >0$. Step (b) follows because for $n > \nu + w$, the test statistic $W_{n,A}$ (defined in \eqref{eq:CuSum_stat}) is greater than or equal to $W_{n,A}^\nu$ (defined in \eqref{eq:CuSum_stat_m}), thus, $T_{b,A} \leq T^\nu_{b,A}$. Also, by definition (see \eqref{eq:Tbam}), $T^\nu_{b,A} > \nu +w-1$, which means that the positive part can be dropped. Step (c)  is due to the law of iterated expectation (and as explained at the end of Section~\ref{sec:prob_form}, $A_{[\nu,\nu+w-1]}$ denotes the set $\{A_{\nu}, \ldots, A_{\nu+w-1}\}$). Step (d) holds because  
$T_{b,A}^{\nu}$ depends on the data before time $\nu$ only through the controls $A_{[\nu,\nu+w-1]}$, which determine the distributions of the data $X_{[\nu+1,\nu+w]}$. Finally, step (e) holds due to the definition of $T_{b,A}^{\nu}$ and the fact that $T_{b,A}= T_{b,A}^{\nu}$ for $\nu=1$. 
\end{proof}

Next, we establish an auxiliary result for the  MLE, which  applies to any control policy that samples uniformly at random from the set $\bA$ at the subsequence of time instances $\cN$.  Specifically, at each time $n> w$, we upper bound the conditional
probability that the MLE at some time instant makes an error in identifying the true post-change parameter $\theta$, given all the already collected data apart from those at the $w$ most recent time instants.

\begin{comment}
\textcolor{red}{
For this bound, we recall the definition of $\rho$ in \eqref{max_Bhat}, and for any $\eta>1$ we set 
\begin{equation} \label{def:r}
r_\eta := \frac{-\log \rho}{\log \eta}.
\end{equation}
Since $\rho<1$, we clearly have $r_\eta \in (0, \infty)$ for any $\eta>1$. }
\end{comment}

\begin{lemma} \label{lem:cons}
If $A \in \cA$  samples uniformly at random from $\bA$ at the subsequence of time instances $\cN$, then:
\begin{equation} \label{eq:corr}
\Pro_{1,A}^\theta \{ \hat{\theta}_n \neq \theta \, |\,  \cF_{n-w-1}, A_{[1,w]} \} \leq |\Theta-1| \, \rho^{q}, \quad \forall \, n >w.
\end{equation}

\end{lemma}

\begin{proof}
Fix an arbitrary integer  $n$ with $n>w$.
From the definition of the MLE in \eqref{eq:ML} it is clear that  $\hat{\theta}_n = \theta$ if $S_n (\theta, \tilth) > 0$  for all $\tilth \neq \theta$, 
where 
\begin{align*}
    S_n (\theta, \tilth) &:= \sum_{m=n-w}^{n-1}  Z_m (\theta, \tilth), \quad Z_m (\theta, \tilth) :=  \log \left( \frac{f_{A_m}^\theta (X_n)}{f_{A_m}^{\tilth} (X_n)}\right).
\end{align*}
Therefore, 
\begin{align*}
\Pro_{1,A}^\theta \{\hat{\theta}_n \neq \theta  \, |\, \cF_{n-w-1}, \,  A_{[1,w]}\} & = \Pro_{1,A}^\theta \left( \bigcup_{\tilth \neq \theta} \{ S_n (\theta, \tilth) \leq 0\}  \, \Bigg|\, \cF_{n-w-1},  \,   A_{[1,w]}\right),
\end{align*}
and  by an application of  the union bound we obtain 
\begin{align}
\Pro_{1,A}^\theta \{\hat{\theta}_n \neq \theta  \, |\,  \cF_{n-w-1}, \,  A_{[1,w]}\}
& \leq \sum_{\tilth \neq \theta}  \Pro_{1,A}^\theta\{S_n (\theta, \tilth) \leq 0  \, |\, \cF_{n-w-1}, \,  A_{[1,w]}
  \} \label{eq:UB}.
\end{align}
Furthermore, by  the conditional Markov inequality it follows that:
\begin{align} \label{eq:MEUB}
\begin{split}
\Pro_{1,A}^\theta\{S_n (\theta, \tilth) \leq 0  \, |\, \cF_{n-w-1},  A_{[1,w]} \}
& \leq \Exp_{1,A}^{\theta} \left[ \exp\left\{-\frac{1}{2} S_n (\theta, \tilth)\right\}  \, \bigg|\,  \cF_{n-w-1}, \, 
 A_{[1,w]} \right]. \\
%&= \Exp_{1,A}^{\theta} \left[\prod_{k=1}^n e^{-\frac{1}{2} Z_k (\theta, \tilth)}\right]. \label{eq:MEUB}
\end{split}
\end{align}

By the conditional independence of $Z_m (\theta, \tilth)$ given $\cF_{m-1}$ and \eqref{bhat}-\eqref{bhat_average} we have
\begin{align*}
&\Exp_{1,A}^{\theta} \left[ \exp\left\{-\frac{1}{2} Z_m (\theta, \tilth)\right\} \, \Bigg| \, \cF_{m-1}\right]\\
& =\Exp_{1,A}^{\theta} \left[  \sqrt{\left( \frac{f_{A_m}^{\tilth} (X_n)}{f_{A_m}^{\theta} (X_n)}\right) } 
\, \Bigg| \, \cF_{m-1} \right]
= \begin{cases}
\rho (\theta, \tilth) &  \text{if $m \in \cN$}\\
\rho (\theta, \tilth, A_m) & \text{otherwise}
\end{cases}\\
& \leq \; \rho^{\mathbb{I}\{m \in \cN\} },
\end{align*}
where $\rho$ is defined in \eqref{max_Bhat} and $\mathbb{I}$ is the indicator function. Applying the  the previous inequality for $m=n-1$  we obtain
\begin{align*}
\Exp_{1,A}^{\theta} \left[ e^{-\frac{1}{2} S_n (\theta, \tilth)}  \, |\,
\cF_{n-w-1},\, A_{[1,w]}\right]
& = \Exp_{1,A}^{\theta} \left[ e^{-\frac{1}{2} S_{n-1} (\theta, \tilth)} \: \Exp_1^{\theta} \left[ e^{-\frac{1}{2} Z_{n-1} (\theta, \tilth)} \Big| \cF_{n-2}\right]  \, \Big|\, 
\cF_{n-w-1}, \, A_{[1,w]} \right]\\
&\leq  \Exp_{1,A}^{\theta} \left[ e^{-\frac{1}{2} S_{n-1} (\theta, \tilth)} \:   \rho^{\mathbb{I}\{(n-1) \in \cN\}}  \, \Big|\, \cF_{n-w-1}, \,  A_{[1,w]} \right] .
\end{align*}
Here, the equality follows by an application of the law of iterated expectation, since $w\geq 1$ implies that $\cF_{n-2} \supseteq \cF_{n-w-1}$, and $n >w$ implies that $A_{[1,w]}$, which is $\cF_{w-1}$-measurable,  is also $\cF_{n-2}$-measurable. 

Repeating the same argument $w$ times  we obtain:
\begin{align}
\begin{split}
\Exp_{1,A}^{\theta} \left[ \exp\left\{-\frac{1}{2} S_n (\theta, \tilth)\right\}  \, \Big|\, \cF_{n-w-1}, \,  A_{[1,w]}   \right] & \leq  \rho^{ \sum_{m=n-w}^{n-1} \mathbb{I} \{ m \in \cN\} } \\
& \leq  \rho^{q}.
\label{eq:B_bound}
\end{split}
\end{align}
where the second inequality follows from \eqref{sampling}, according to which at any time-period of length $w$ there are at most $q$ times in $\cN$. Combining \eqref{eq:B_bound}, along with \eqref{eq:UB} and \eqref{eq:MEUB} completes the proof.
\end{proof}

We next  establish a lower bound on the drift of the log-likelihood ratio process that drives the WCC procedure in the post-change regime. This lower bound will be useful in upper bounding the delay of the WCC procedure. For this, we need to introduce some additional notation. Thus, 
for every $u,\theta\in \Theta$, $a\in \mathbb{A}$, we set 
\begin{align}
B_a(\theta,u) & \equiv 
\int \log (f_a^u / f_a^{\theta_0}) f_a^\theta d \lambda \label{eq: Badef}\\
&= D(f_a^\theta\| f_a^{\theta_0}) - D(f_a^\theta\| f_a^u) = I_a^\theta - D(f_a^\theta\| f_a^u)
\nonumber.
\end{align}
Note that
\[
B_a(\theta,\theta) = I_a^\theta.
\]
Also, let
\begin{equation} \label{eq:astaru}
a^* (u) \in \arg \max_{a\in \mathbb{A}} I_a^u = \arg \max_{a\in \mathbb{A}}D(f_a^u\| f_a^{\theta_0}).
\end{equation}
Note that $a^* (u)$ is an optimal control for the WCC procedure at time $n \notin \cN$, when the post-change parameter estimate $\hat{\theta}_n=u$ (see \eqref{eq:Anstar}). Also, it is clear that:
\begin{equation}\label{eq:Bastar}
B_{a^* (\theta)} (\theta,\theta) = I_{a^*(\theta)}^\theta = \max_{a\in \mathbb{A}} I_a^\theta = I^\theta.
\end{equation}

\begin{lemma} \label{lemma:LB on Exp}
Let $n>w$. Then,
\begin{align}
\Exp_{1,A^*}^\theta[\Lambda_{n,A^*}^{\hat{\theta}_n} \, | \,  \cF_{n-w-1}, \,   A_{[1,w]}^* ]\geq
\begin{cases}
I^\theta- J^\theta  \, \rho^{q}, & \quad n \notin \cN \\
\frac{1}{|\mathbb{A}|} \sum_{a \in \mathbb{A}} (I_a^\theta - K_a^\theta \rho^q),  &\quad n \in \cN \\
\end{cases} \label{eq:lem4}
\end{align}
where 
\begin{align*}
J^\theta &\equiv   (|\Theta|-1) \cdot (I^\theta- \min_{u \neq \theta} B_{a^*(u)}(\theta,u) ), \quad \\
K_a^\theta &\equiv   (|\Theta|-1) \cdot (I_a^\theta- \min_{u \neq \theta} B_{a}(\theta,u) ),
\end{align*}
with $a^*(u)$ as defined in \eqref{eq:astaru}.
Moreover, for large enough $q$, both  lower bounds in \eqref{eq:lem4} are strictly positive.
\end{lemma}

\begin{proof}
When  $n \notin \cN$, $A_n^*$ is chosen according to \eqref{eq:Anstar}, and 
\begin{align*}
&\Exp_{1,A^*}^\theta[\Lambda_{n,A^*}^{\hat{\theta}_n} \, | \, \cF_{n-w-1}, \,  A_{[1,w]}^* ] \\
&= \sum_{u \in \Theta}  
\Exp_{1,A^*}^\theta\left[\log \frac{f_{A_n^*}^{\hat{\theta}_n} (X_n)}{f_{A_n^*}^{\theta_0} (X_n)} \, \Bigg| \, \hat{\theta}_n=u, \, \cF_{n-w-1}, \,  A_{[1,w]}^* \right]
\cdot \Pro^\theta_{1,A^*} (\hat{\theta}_n=u \, |\, \cF_{n-w-1}, \, A_{[1,w]}^* )\\
& \stackrel{\text{(a)}}{=} \sum_{u \in \Theta} \left( \int \log \left(\frac{f_{a^*(u)}^u}{f_{a^*(u)}^{\theta_0}} \right) \, f_{a^*(u)}^\theta \,  d \lambda \right) \cdot \Pro^\theta_{1,A^*} (\hat{\theta}_n=u \, |\, \cF_{n-w-1}, \,   A_{[1,w]}^* )\\
&= \sum_{u \in \Theta} B_{a^*(u)}(\theta,u)
\cdot \Pro^\theta_{1,A^*} (\hat{\theta}_n=u \, |\, \cF_{n-w-1}, \,  A_{[1,w]}^* )\\
&\stackrel{\text{(b)}}{\geq}  I^\theta \cdot  \Pro_{1,A^*}^\theta (\hat{\theta}_n = \theta \, | \, \cF_{n-w-1} \, , \,   A_{[1,w]}^* )  + \min_{u \neq \theta} B_{a^*(u)}(\theta,u) \cdot \Pro_{1,A^*}^\theta (\hat{\theta}_n \neq \theta 
\, | \,   \cF_{n-w-1} \, , \,    A_{[1,w]}^* ) \\
&=  I^\theta - \left(I^\theta-\min_{u \neq \theta} B_{a^*(u)}(\theta,u)\right) \cdot  \Pro_{1,A^*}^\theta (\hat{\theta}_n \neq \theta \, | \,  \, \cF_{n-w-1} \, , \,   A_{[1,w]}^* ) \\
&\stackrel{\text{(c)}}{\geq}   I^\theta - \left(I^\theta-\min_{u \neq \theta} B_{a^*(u)}(\theta,u)\right)  \cdot (|\Theta|-1) \cdot \rho^{q},
\end{align*}
where (a) follows from \eqref{eq:Anstar} and \eqref{eq:astaru}, and from the fact that conditioned on $\hat{\theta}_n =u$, the optimal control $A_n^*= a^*(u)$ and $X_n$ is independent of everything else. Step (b) follows from \eqref{eq:Bastar}, and (c) from Lemma~\ref{lem:cons}. This proves the first inequality
in \eqref{eq:lem4}.

When $n \in \cN$, $A_n^*$ is chosen uniformly at random from the set $\mathbb{A}$, and following steps similar to those used in the case when $n\notin \cN$, we obtain:
\begin{align*}
&\Exp_{1,A^*}^\theta[\Lambda_{n,A^*}^{\hat{\theta}_n} \, | \, \cF_{n-w-1} \, ,   \,  A_{[1,w]}^* ] \\
&= \sum_{u \in \Theta}  
\Exp_{1,A^*}^\theta \left[\log  \left( \frac{f_{A_n^*}^{\hat{\theta}_n} (X_n)}{f_{A_n^*}^{\theta_0} (X_n)} \right) \, \Bigg| \, \hat{\theta}_n=u, \, \cF_{n-w-1} \, ,   \,   A_{[1,w]}^* \right]
\cdot \Pro^\theta_{1,A^*} (\hat{\theta}_n=u \, |\, \cF_{n-w-1} \, ,   \,   A_{[1,w]}^* )\\
&=  \sum_{u \in \Theta} \frac{1}{|\mathbb{A}|} \sum_{a \in \mathbb{A}} B_a(\theta,u) \cdot \Pro^\theta_{1,A^*} (\hat{\theta}_n=u  \, | \,  \, \cF_{n-w-1}\, ,   \,    A_{[1,w]}^*) \\
&= \frac{1}{|\mathbb{A}|} \sum_{a \in \mathbb{A}} \sum_{u \in \Theta}
B_a(\theta,u) \cdot \Pro^\theta_{1,A^*} (\hat{\theta}_n=u  \, | \, \cF_{n-w-1} \, ,   \, A_{[1,w]}^*) \\
&\geq  
\frac{1}{|\mathbb{A}|} \sum_{a \in \mathbb{A}} \bigg[
I_a^\theta \cdot  \Pro_{1,A^*}^\theta (\hat{\theta}_n = \theta \, | \,  \cF_{n-w-1}\, ,   \,    A_{[1,w]}^* )  \\
& \hspace*{1.5in} + \min_{u \neq \theta} B_{a}(\theta,u) \cdot \Pro_{1,A^*}^\theta (\hat{\theta}_n \neq \theta 
\, | \,    \cF_{n-w-1} \, ,   \,  A_{[1,w]}^* ) \bigg]\\
& \geq 
\frac{1}{|\mathbb{A}|} \sum_{a \in \mathbb{A}} \left[I_a^\theta - \left(I_a^\theta-\min_{u \neq \theta} B_{a}(\theta,u)\right)  \cdot (|\Theta|-1) \cdot \rho^{q}\right],
\end{align*}
which proves the second inequality
in \eqref{eq:lem4}. 

Finally, we note  that from \eqref{eq: Badef} it follows that, for all $a \in \mathbb{A}$ and  $u \in \Theta$,
\[
B_a (\theta,u) \leq I_a^\theta \quad \text{and}\quad B_{a^*(u)} (\theta,u) \leq I^\theta.
\] Therefore,  $J^\theta \geq 0$ and $K_a^\theta \geq 0$, for all $a \in \mathbb{A}$. Nevertheless,  since $0 < \rho < 1$, as $q\to \infty$ we have: 
\[
I^\theta- J^\theta \, \rho^{q} \to I^\theta~ > ~ 0
\]
and 
\[
\frac{1}{|\mathbb{A}|} \sum_{a \in \mathbb{A}} (I_a^\theta - K_a^\theta \, \rho^q) \to \frac{1}{|\mathbb{A}|} \sum_{a \in \mathbb{A}} I_a^\theta ~ > ~ 0,
%~< ~I^\theta.
\]
which proves that both lower bounds in \eqref{eq:lem4} are  both positive for large enough $q>0$.
\end{proof}

\begin{theorem} \label{th:aub}
As $q, w, b \to \infty$ so that $q=o(w)$ and  $w =o(b)$, 
$$\sup_{a_1, \ldots, a_w 
\in \bA} \Exp_{1,A^*}^\theta[T_{b,A^*} \,|\, A^*_1=a_1, \ldots, A^*_w=a_w] \leq \frac{b}{I^\theta}\, (1+o(1)).$$
\end{theorem}

\begin{proof}

Define the stopping time
\begin{equation*}
T'_{b,A^*}:= \inf\left\{n > w:  W'_{n,A^*}  \geq b\right\}.
\end{equation*}
with
\begin{align} 
W'_{n,A^*} &=  W'_{n-1,A^*} + \Lambda^{\hat{\theta}_{n}}_{n,A^*}, \quad n > w, \label{eq:CuSum_stat_2}\\
W'_{w,A^*}&=0.\nonumber
\end{align}
Note from \eqref{eq:CuSum_stat} that the difference between $W'_{n,A^*}$ and $W_{n,A^*}$ is that we do not take the positive part on the right-hand-side of the recursion in \eqref{eq:CuSum_stat_2}. This means that $T_{b,A^*}\leq T'_{b,A^*}$. Therefore, it suffices to show that as $q, w, b \to \infty$ so that $q=o(w)$ and  $w =o(b)$, 
$$
\sup_{a_1, \ldots, a_w 
\in \bA} \Exp_{1,A^*}^\theta[T'_{b,A^*} \,|\, A^*_1=a_1, \ldots, A^*_w=a_w] \leq \frac{b}{I^\theta} \, (1+o(1)).$$
In what follows, we fix arbitrary $a_1, \ldots, a_w$ in $\bA$, and we show that 
$$\Exp_{1,A^*}^\theta[T'_{b,A^*} \,|\, A^*_1=a_1, \ldots, A^*_w=a_w] \leq \frac{b}{I^\theta} \, (1+o(1)).$$
To see this, we  make the following identifications in Proposition~\ref{th:appen} in the Appendix:
\begin{align*}
\Pro \leftrightarrow 
\Pro_{1,A^*}^\theta (\cdot \; | \; A^*_1=a_1, \ldots, A^*_w=a_w) , \quad \cG_n \leftrightarrow 
\cF_n, \qquad 
Y_n &\leftrightarrow \Lambda^{\hat{\theta}_n}_{n,A^*},   \quad \tau_b \leftrightarrow T'_{b,A^*} .
\end{align*}
We also observe that by Lemma \ref{lemma:LB on Exp} we can set $\mu_*, \mu^*, \mu, v$ in Proposition~\ref{th:appen} as follows:
\begin{align*}
\mu_* &\leftrightarrow \frac{1}{|\mathbb{A}|} \sum_{a \in \mathbb{A}} (I_a^\theta - K_a^\theta \, \rho^q), \quad 
\mu^* \leftrightarrow I^\theta- J^\theta  \, \rho^{q}, \quad 
\mu \leftrightarrow I^\theta, \quad 
V \leftrightarrow v.
\end{align*}
By Lemma \ref{lemma:LB on Exp} again we have $\mu^* \to  I^\theta=\mu$ as $q \to \infty$.
%and observing that   by \eqref{bhat_average}-\eqref{max_Bhat}we have  $\rho \in (0,1)$, which  implies that  $\rho^q \to 0$ as $q \to \infty$.  Therefore, $\mu^*$ and $\mu_*$ are ??? both positive ??? for large enough $q$,and, also,  $\mu^* \to I^\theta=\mu$ as $q \to \infty$. 
Thus, the condition in part (iii) of Proposition~\ref{th:appen} is satisfied, which implies the result. 

%\textcolor{red}{Venu: We need to update this proof with the new Lemma 4 and reference Lemma 4 in the proof.}

%\textcolor{red}{Venu: I don't agree with the identification $S_n \rightarrow W_{n,A^*}$.The CuSum statistic is not just the sum of the log-likelihoods. There is also an issue with $\mu_* \rightarrow C, \quad  \mu^* \rightarrow I^\theta- D_\theta  \, \rho^{q}$. These need to be reversed? Also, is there an inequality relationship between $C$ and $I^\theta- D_\theta  \, \rho^{q}$? There appears to be one relating $\mu^*$ and $\mu_*$ with $\mu_* < \mu^*$.}

        \end{proof}\
We then have the following result that establishes the first-order optimality of the WCC procedure.        
        \begin{corollary}
    If $b=\log \gamma$, then as $q,w,b\to \infty$ so that 
$q=o(w)$ and $w =o(\log \gamma)$, then $$\cJ_\theta(A^*,T)\sim 
\inf_{ (A,T) \in \cC_\gamma}  \cJ_\theta(A,T) \sim \frac{\log \gamma}{ I^\theta} 
$$
    \end{corollary}

\begin{proof}
This follows by combining Theorem \ref{th:LB}, Lemma \ref{lem1},  Lemma \ref{lem:worst-case-bd}, and Theorem \ref{th:aub}.  
    \end{proof}

\section{Simulations} \label{sec:sim}
We consider a simulation study in which, for each $n \in \bN$, $X_n$ is a $K$-dimensional random vector $(X_n^1, \ldots, X_n^K)$ with independent components. We assume that each of these components has   variance 1 and that  the mean changes in an unknown subset of components, $\theta$.  Specifically, if $k \notin \theta$, each $X_n^k$ has mean 0, whereas if $k \in \theta$, then $X_n^k$ has mean 0 for $n<\nu$ and mean $\mu_k$ for $n\geq \nu$. Therefore, in this context, the parameter space, $\Theta$, consists of all non-empty subsets of $\{1, \ldots, K\}$. We further assume that it is possible to sample only one of these components at each time instant. Therefore, the action space $\mathbb{A}$ consists of all singletons $\{ \{k\}: 1\leq k \leq K\}$. 

We set $K=10$, %we consider the following  setup for the potential post-change means.
%A homogeneous one, where  $(\mu_1, \ldots, \mu_k)=(1,1,1, 1, 1, \ldots, 1)$ and a heterogeneous one where 
$(\mu_1, \ldots, \mu_k)=(0.5, 0.5, 1, 1, 1, \ldots, 1)$, 
%That is, %in the second setup,if one of the first two components is affected by the change, its mean after the change will be half of that of another sampled after the , the signal strength of the change in these two components will be smaller compared to that of other components that may be affected by the change.  In our simulation studies, 
and we consider the case that only the first three components experience the change, i.e., the true  value of $\theta$ is $\{1,2,3\}$. That is, among the three affected components, only one of them has the largest possible drift. Moreover, we assume that  the  change  occurs at time $\nu=1$. 

We consider two schemes.

\begin{itemize}
%\item  The \textit{oracle} CuSum rule, which only samples  the affected component with the largest post-change mean, i.e.,  component 3.
%in the heterogeneous setup and any of components 1,2,3 in the homogeneous setup)
\item[(i)]  The proposed \textit{Windowed Chernoff-CuSum} procedure, with two choices for the size of the window, $w$, $10$ and $20$. With this procedure, we  sample uniformly at random from all sources at the first $w$ time instants, and an alarm can be raised only after this time. For both window choices, we set $q=1$, that is, during each interval of the form $[m w, (m+1) w)$, where  $m \in \bN$, only once the component is selected at random. Apart from these times, at eact time after $w$, we estimate the  subset of components that have been affected by the change, and  select to sample among those  the one  with the largest post-change mean. This means that  whenever a component apart from the first two is estimated to be affected by the change, it can be selected for sampling next. In order to break such ties,  we select the component whose observations in the last $w$ time instants had the largest average.  This choice was found to lead to much better performance than random sampling among the components that are currently estimated as affected by the change.

%\item The version of the proposed scheme that selects a component completely at random at each time instant, to which we refer as \textit{random sampling}.
\item The \textit{greedy} algorithm considered in \cite{Draga96,zhang-mei2020banditQCD_journal,Mei_Moust_2021,Anamitra2021,XuMeiPostUncert_SeqAn}, according to which we keep sampling the same component until its cumulative log-likelihood ratio statistic    hits  either the threshold $b$ or $0$. In the former case, the alarm is raised and the process is terminated. In the latter, all past observations are forgotten and the next component is sampled.  This algorithm has been shown  to enjoy certain asymptotic optimality properties 
when either only a single component is affected by the change \cite{Mei_Moust_2021} or all affected components have the same post-change distribution \cite{Anamitra2021}, but not in the setup of this  simulation study. Moreover, 
its  performance  depends heavily on the component that is being sampled when the change occurs. In order to capture its  average behavior, 
we simulate 
at random the component that  is being sampled at the time of the change. For comparison purposes, we also consider the  best case/oracle scenario, where  it is component 3, i.e., the one with the largest post-change mean among those affected,  that is being sampled at the time of the change in each simulation run.  
\end{itemize}

In all cases, the threshold $b$ is selected  as $\log \gamma$, where  $\gamma$ takes values in 
$\{10^{i}: i=1,\ldots, 16\}$.  Figure \ref{fig1} presents the results, where in the vertical axis we plot the average detection delays of the above  schemes,  based on $16,000$ simulation runs for each of them, and in the horizontal axis we plot $\log(\gamma)$.  
 
Based on this graph,  we can make the following observations. First, for the proposed Windowed Chernoff-CuSum,  the larger window, i.e., setting $w$ equal to  $20$ instead of $10$, leads to smaller expected detection delay for large enough values of $\gamma$.
%for large values of $\gamma$  to essentially the same expected detection delay in the heterogeneous setup, and even better expected detection delay in the heterogeneous setup. 
This is consistent with our asymptotic theory, according to which the window  size $w$ needs to go to infinity as $\gamma \to \infty$ in order to achieve asymptotic optimality. Second, both these window choices  lead to expected detection delays  smaller than that  of the greedy algorithm. In fact,  the expected detection delay of the proposed scheme, with the larger window, approaches the oracle  for the greedy  algorithm as $\gamma$ increases. 

%the difference between these two window choices is rather small, and both of them  lead to much better than pure random sampling, especially as $\gamma$ increases. 
%On the other hand, this is not the case for  greedy algorithm, whose performance depends critically on the component that is sampled at the time of the change. In the best case scenario, where at the time of the change the component that is sampled is affected by the change and has the maximum possible drift, the expected detection delay of this algorithm strictly outperforms the proposed one and approaches the performance of the oracle in the homogeneous setup. In the heterogeneous setup, the  proposed scheme with the larger window attains the same performance for large values of $\gamma$. However, in both setups, even the \textit{average} (with respect to the component that is sampled at the time of the change) expected detection delay is worse than even that of the ``naive'', random sampling scheme. 

\begin{figure}
\centering
\includegraphics[width=0.8\linewidth]{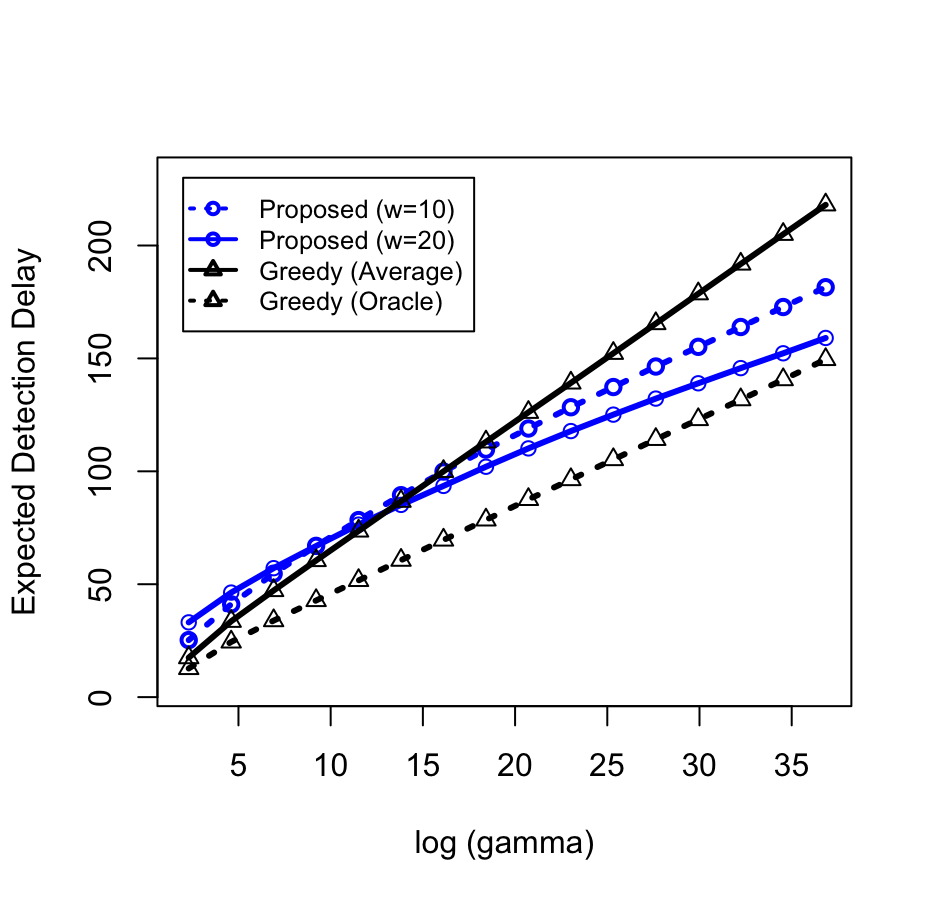} 
\caption{The vertical axis represents Expected Detection Delay and the horizontal axis  $\log \gamma$. The blue lines/circles correspond to the proposed, window-limited Chernoff-CuSum scheme with windows $w=10$ (dashed line) and $w=20$ (solid line). The black lines/triangles correspond to the greedy algorithm. The solid line corresponds to its average behavior, and the dashed line to an idealistic case.}
\label{fig1}
\end{figure}

%We observe that the scheme that utilizes random sampling  perform significantly worse. Moreover, we observe that whereas for small values of $\gamma$, the smaller value of $w$ is  preferable, the situation is reversed  for larger values of $\gamma$. 

\section{Conclusions} \label{sec:concl}
We considered the problem of quickest change detection when there is parameter  uncertainty regarding the post-change regime,  and 
the distribution of the observations is affected not only by the change but also by a control input that is chosen by the observer. We provided a precise mathematical formulation for a general setting of this problem. We then developed a procedure that is asymptotically optimal under Lorden's criterion as the false alarm rate goes to zero. 

Our asymptotic optimality result is established under the assumption of a finite action space and a finite parameter space regarding the post-change distribution. Therefore, a natural direction for further research is the theoretical analysis of this algorithm in the case of an infinite action space and an infinite post-change parameter space. Another direction  for further research is the consideration of more general stochastic models for the observations, in which the assumption of conditional independence is removed, in the spirit of \cite{niti-veer-sqa-2015}. Also, in our problem formulation, we implicity assumed that all control actions incur the same cost. Extending the formulation and the analysis to allow for non-uniform costs across control actions would be of interest in certain applications. Finally, the proposed algorithm relies on a window of past observations during which the post-change parameter is estimated, and this is subsequently used to assign the appropriate control. Our asymptotic theory requires that the length of this window goes to infinity as the false alarm rate goes to 0. An interesting question is to obtain more precise bounds for the worst-case conditional expected detection delay that could inform the selection of this window size. 

%Of course, such an optimal choice will depend on the true, but unknown, post-change density. This suggests that another interesting direction for further study is the consideration of an adaptive choice for this window. 

%[Need to talk about insights from numerical results, and provide some avenues for future research. For example, what happens when the parameter and action sets are continuous?]

%\section*{Acknowledgements}The authors would like to thank George Moustakides for helpful discussions and for suggesting the use of the windowed CuSum statistic in constructing the WCC procedure.

\appendix

\section{An auxillary lemma}

%\textcolor{red}{[Venu: Need to say something about Theorem 3 being of independent interest; should this be called a lemma instead? Also, we should probably move some of the other proofs to the Appendix ]}

In this Appendix we present a general, non-asymptotic upper bound for the expected hitting time of a sequence for which  the conditional expectation of each increment given all past observations does not exceed a certain positive constant. This is the basis of the proof of Theorem \ref{th:aub}, but it is stated in a more abstract form, as it can be of  independent interest. \\

\begin{prop} \label{th:appen}
Let $(\Omega, \cF, \Pro)$ be a probability space, and let $\Exp$ denote expectation with respect to $\Pro$. Let  $\{\cG_n, n \geq 0\}$ be a filtration on this space, where $\cG_0$ is the trivial $\sigma$-algebra.  Let 
$\{Y_n, n \geq 1\}$  a  sequence of random variables in $\cL^2$  that is adapted to $\{\cG_n, n\geq 1\}$. 
Let $\omega \in \bN$, $\cN \subseteq \{w+1, w+2, \ldots\}$,  $  \mu_*,  \mu^*, \mu, v>0$, and  $q \in \bN$, and suppose that, for every $n>w$,
\begin{align}\label{mu_low}
\Exp[Y_n   \,|\,  \cG_{n-w-1} ] &\geq \begin{cases}
 \mu^*, \quad   n \notin \cN \\
 \mu_*,  \quad   n \in \cN,
\end{cases} \quad  \text{a.s.}
\end{align}
\begin{align}
 \Exp[Y_n \, |\, \cG_{n-1}] &\leq \mu \quad  \text{a.s.} \label{mu_up}\\
\Exp[Y^2_n \, | \, \cG_{n-1}] &\leq v \quad  \text{a.s.}
\label{variance}
\end{align}
and 
\begin{align} \label{sampling2}
\sum_{i=w+1}^{w+n}  \mathbb{I}\{i \in \cN\}\leq  \frac{n}{w} \,  q.
 \end{align}
 %i.e., there are at most $q$ elements of $\cN$ in any interval of length $w$.
For any $b>0$, set 
$$\tau_{b} \equiv \inf\{ n >  w: S_n \geq b\}, \quad \text{where} \quad S_n \equiv \sum_{i=w+1}^n Y_i.
$$
Then:
\begin{itemize}
\item[(i)]  $\Exp[\tau_b]<\infty$.
\item[(ii)]  There is a $C>0$, that does not depend on $b$, $q$, or $w$, so that
$$
\Exp[\tau_b ] \leq \frac{b+w\,  \mu + C (1+\sqrt{b}+\sqrt{w})}{  \mu^* (1-q/w)}.
$$
\item[(iii)] If $\mu^* \to \mu$  as $q,w \to \infty$ so that $q=o(w)$, then 
$$
\Exp[\tau_b] \leq \frac{b}{\mu}\, (1+o(1)),
$$
where $o(1)$ is a vanishing term as 
 $q, w, b \to \infty$ so that $q=o(w)$ and  $w=o(b)$.
\end{itemize}
\end{prop}

\begin{remark}
This proposition  generalizes    \cite[Proposition 2.1]{Lalley_Lorden_1986}, as  in the latter condition \eqref{mu_low} is assumed to hold for $w=1$  and condition  \eqref{sampling2} for $q=0$. 
\end{remark}

  \begin{proof}
(i)  For any  $c>0$ and $n >w$ we have 
\begin{align*}
\Exp\left[Y_n \cdot  \mathbb{I}\{Y_n \geq c\} \,  
\, | \, 
\cG_{n-w-1
} \right]
&\overset{(\text{a})} \leq \sqrt{\Exp[Y^2_n \,  
\, | \, \cG_{n-w-1}]} \; \sqrt{  \Pro(Y_n \geq c  
\, | \, \cG_{n-w-1}}) \\
&\overset{(\text{b})}\leq  \frac{\Exp\left[Y_n^2
\, | \, \cG_{n-w-1}
\right]}
{c} \overset{(\text{c})}\leq \frac{v}{c} \quad \text{a.s.},
\end{align*}
where (a) follows by the conditional Cauchy-Schwarz inequality, (b) by the  
 conditional Markov inequality, (c) by the law of iterated expectation and \eqref{variance}.
As a result,  for $c$ large enough we have 
\begin{equation}
\label{aux}
\Exp[Y_n \cdot \mathbb{I}\{Y_n \geq c\} \, | \, \cG_{n-w-1}]\leq (\mu_* \wedge  \mu^*)/2 \quad \text{a.s.},
\end{equation}
where we use $a\wedge b$ to denote $\min\{a,b\}$. Consequently, for $c$ large enough we have 
\begin{align}
\Exp[Y_n  \wedge c \, | \, \cG_{n-w-1}]  &\geq  \Exp[Y_n \cdot \mathbb{I}\{Y_n < c\} \, | \, \cG_{n-w-1} ] \nonumber\\
&= \Exp[Y_n \, | \, \cG_{n-w-1}]-  \Exp[Y_n \cdot \mathbb{I}\{Y_n \geq  c\} \, | \, \cG_{n-w-1} ] \nonumber\\
& \geq \mu_* \wedge  \mu^* - (\mu_* \wedge  \mu^*)/2= (\mu_* \wedge  \mu^*)/2 \quad \text{a.s.},
\label{eq:LBYnminc}
\end{align}
where the last inequality follows from \eqref{mu_low} and \eqref{aux}.
For every $n \in \bN$ with $n >w$  we set 
\[
\tilde{Y}_n \equiv Y_n \wedge c \qquad \text{and} \qquad  \tilde{S}_n\equiv \tilde{Y}_{w+1}+\ldots +\tilde{Y}_n,
\]
thus,  by \eqref{eq:LBYnminc}, we have 
\begin{equation}
\Exp[\tilde{Y}_n \, | \, \cG_{n-w-1}] \geq (\mu_* \wedge  \mu^*)/2 \quad \text{a.s.}
\label{eq:LBYnminc2}
\end{equation}
For every $b>0$ we set 
\begin{align*}
\tilde{\tau} \equiv \inf\{n >w: \tilde{S}_n \geq b\}. 
\end{align*}
Note that $\tau_b \leq \tilde{\tau}$ almost surely. 
Let $m \in \bN$ with $m >w$. Then:
 \begin{align}\label{aux0}
 \tilde{S}_{\tilde{\tau} \wedge m} \leq b+c.
 \end{align}
  Indeed, if   $\tilde{\tau}> m$, then $\tilde{S}_{\tilde{\tau} \wedge m}=\tilde{S}_m<b$, whereas  if $\tilde{\tau}\leq m$, then 
\begin{align*}
\tilde{S}_{\tilde{\tau} \wedge m}= \tilde{S}_{\tilde{\tau}}
\leq  b+\tilde{Y}_{\tilde{\tau}} \leq  b+c.
\end{align*}
Moreover,  
\begin{align}\label{aux05}
\tilde{S}_{\tilde{\tau}  \wedge m} &=\sum_{n=w+1}^{\tilde{\tau}\wedge m+w} \tilde{Y}_n- \sum_{n=\tilde{\tau}\wedge m+1}^{\tilde{\tau}\wedge m+w} \tilde{Y}_n.
\end{align}
We start by lower bounding the conditional expectation of the first term in the right-hand side in \eqref{aux05}. To this end, we first observe that 
\begin{align}\label{aux1}
\begin{split}
\Exp\left[\sum_{n=w+1}^{\tilde{\tau}\wedge m+w} \tilde{Y}_n  \right] &= \Exp\left[\sum_{n=w+1}^{\infty} \tilde{Y}_{n} \cdot \mathbb{I}\{\tilde{\tau} \wedge m  +w\geq n\}  \right]\\
&\stackrel{\text{(a)}}{=} \sum_{n=w+1}^{\infty} \Exp\left[ \tilde{Y}_{n} \cdot \mathbb{I}\{\tilde{\tau} \wedge m  +w\geq n\} \right]\\
&\stackrel{\text{(b)}}{=} \sum_{n=w+1}^{\infty} \Exp\left[ \Exp\left[ \tilde{Y}_{n}  \, |\, \, | \, \cG_{n-w-1} \right] \cdot \mathbb{I}\{\tilde{\tau} \wedge m+w \geq n\} \right].
\end{split}
\end{align}
The interchange of series and expectation in step (a) is justified because $\tilde{\tau} \wedge m$ is a bounded stopping time and the sequence $(\tilde{Y}_n)$ is  bounded above.  Step (b) follows by an application of the law of iterated expectation, in view of the fact that 
$\{\tilde{\tau} \wedge m +w \geq n\}$ is $\cF_{n-w-1}$-measurable for  $n>w$.

Then, by \eqref{aux1} and \eqref{eq:LBYnminc2} we conclude that 
\begin{align}\label{aux11}
\begin{split}
\Exp\left[\sum_{n=w+1}^{\tilde{\tau}\wedge m+w} \tilde{Y}_n  \right] \geq  \frac{\mu_* \wedge \mu^*} {2} \, \Exp\left[\sum_{n=w+1}^{\tilde{\tau} \wedge m +w} 1 \right]  =  \frac{\mu_* \wedge \mu^*} {2} \,  \Exp[\tilde{\tau}   \wedge m ].
\end{split}
\end{align}

We continue by upper bounding the conditional expectation of the second term in the right-hand side in \eqref{aux05}. Since $(\tilde{Y}_n)$ is a bounded  above sequence and $\tilde{\tau} \wedge m$ is a bounded stopping time that takes values larger than $w$, by Wald's identity we have 
\begin{align}\label{aux2}
\Exp\left[\sum_{n=\tilde{\tau} \wedge m+1}^{\tilde{\tau} \wedge m+w+1} \tilde{Y}_n  \right] &=\Exp\left[\sum_{n=\tilde{\tau} \wedge m+1}^{\tilde{\tau} \wedge m+w+1} \Exp[\tilde{Y}_n \, |\, \cG_{n-1}]   \right] \leq \mu  \, w,
\end{align}
where the inequality follows from \eqref{mu_up}.

Combining   \eqref{aux0}, \eqref{aux05}, \eqref{aux11}, and \eqref{aux2}, we conclude that for every $m \in \bN$  with $m >w$ we have 
$$
\Exp[\tilde{\tau}  \wedge m  ] \leq \frac{b+c+w \mu}{  (\mu_* \wedge \mu^*) /2}.
$$
Letting $m \to \infty$, by the Monotone Convergence Theorem we obtain 
\begin{equation} \label{bound}
\Exp[\tau_b  ] \leq \Exp[\tilde{\tau}     ] \leq \frac{b+c+w \mu}{ (\mu_* \wedge \mu^*) /2}<\infty.
\end{equation} 

(ii) We observe that 
\begin{align}\label{wald01}
\begin{split}
\Exp\left[\sum_{n=w+1}^{\tau_b+w} Y_n \right] & \stackrel{\text{(a)}}{=}
 \Exp\left[\sum_{n=w+1}^{\infty}  \Exp[Y_{n} \, |\, \cG_{n-w-1} ] \cdot \mathbb{I}\{\tau_b +w\geq n\}   \right]  \\
&\stackrel{\text{(b)}} \geq \Exp\left[\sum_{n=w+1}^{\infty}  \Exp[Y_{n} \, |\, \cG_{n-w-1} ] \cdot \mathbb{I}\{n \notin \cN\} \cdot \mathbb{I}\{\tau_b  +w\geq n\}    \right]  \\
&\stackrel{\text{(c)}}{\geq}  \mu^* \, \Exp\left[\sum_{n=w+1}^{\tau_b  +w}   \mathbb{I}\{n \notin \cN\} \right]  \\
& \stackrel{\text{(d)}}{\geq}  \mu^*  \,  (1-q/w) \, \Exp[\tau_b ].
\end{split}
\end{align}
Step (a) is obtained by following exactly the same steps as in \eqref{aux1},  with $\tilde{Y}_n$ replaced by $Y_n$ and $\tilde{\tau} \wedge m$ replaced by $\tau_b$, and it is justified because, by the result of part (i), $\tau_b$ is integrable, 
 and, by assumption \eqref{variance},$(\Exp[Y^2_n \, |\, \cG_{n-1}])$ is bounded. Step (b) and (c) follow by \eqref{mu_low}, and step (d) by \eqref{sampling2}.

 Moreover, 
 \begin{align}  \label{wald1}
\begin{split}
  \Exp\left[ \sum_{n=w+1}^{w+\tau_b} Y_n   \right] 
 &=  \Exp[ S_{\tau_b}  \, |\, \cF_{w} ] + \Exp\left[\sum_{n=\tau_b+1}^{w+\tau_b} Y_n  \right]  \\
 &\leq 
 b+  \Exp[ R_{b}  ] +  w \mu,
 \end{split}
 \end{align}
 where  $R_b \equiv S_{\tau_b}- b$ and the inequality follows in exactly the same way as 
  \eqref{aux2}, again with $\tilde{Y}_n$ replaced by $Y_n$ and $\tilde{\tau} \wedge m$ replaced by $\tau_b$. From \eqref{wald01} and \eqref{wald1} we have 
  \begin{align} \label{wald2}
  \Exp[\tau_b]
 &\leq  \frac{ b+ \Exp[R_b  ] +w \mu}{ \mu^* (1 -q/w)},
  \end{align}
and  it  is clear that it suffices to show that there is a $C>0$ so that 
\begin{align} \label{show}
\Exp[R_b ] &\leq C(1 +\sqrt{b} + \sqrt{w}).
\end{align} 
In order to do so, we first claim that 
\begin{align} \label{show00}
0\leq 2 \int_0^b R_c \, dc \leq  \sum_{n=w+1}^{\tau_b} Y_n^2- R_b^2,
\end{align} 
where  $R_c$ is defined as in $R_b$ with $b$  replaced by $c$. To prove this claim, we introduce the sequence of ladder variables
\begin{align*}
\rho_m &\equiv \inf\{n > \rho_{m-1}: S_{n}> S_{\rho_{m-1}}\}, \quad m \in \bN,\\
\rho_0 & \equiv \inf\{n > w: S_{n}> 0\}.
\end{align*}
These can be shown to be integrable, in the same way as we show  the integrability of $\tau_b$ in (i). Moreover, we set 
$$M_b \equiv \inf\{m \in \bN: S_{\rho_m} \geq b\},$$
thus, $\tau_b=\rho_{M_b}$ and 
$R_b=S_{\rho_{M_b}}-b$. Then, 
as in \cite[equation (1)]{Lorden_1970}, we have 
\begin{align} 
0\leq 2 \int_0^b R_c \, dc 
&=  \sum_{m=1}^{M_b} (S_{\rho_m}-S_{\rho_{m-1}})^2- R_b^2. \label{eq:ineqRc1}
\end{align}
By the definition of ladder times, we have  
\[
0< S_{\rho_m}-S_{\rho_{m-1}} \leq Y^+_{\rho_m}, \quad \forall \; m \in \bN.
\]
Therefore
\begin{align} 
  \sum_{m=1}^{M_b} (S_{\rho_m}-S_{\rho_{m-1}})^2
&\leq   \sum_{m=1}^{M_b} Y_{\rho_m}^2
\leq  \sum_{m=1}^{M_b} \sum_{n=\rho_{m-1}+1}^{\rho_m} Y_n^2  = \sum_{n=w+1}^{\tau_b} Y_n^2, \label{eq:ineqRc2}
\end{align}
where the equality holds because 
\[
\tau_b=\rho_{M_b}= \sum_{m=1}^{M_b} (\rho_m-\rho_{m-1}).
\]
Combining \eqref{eq:ineqRc1} and \eqref{eq:ineqRc2} proves \eqref{show00}.

Then, taking expectations  in \eqref{show00} 
and applying Jensen's inequality,
we have 
\begin{align} \label{wald21}
0\leq \Exp \left[ \sum_{n=w+1}^{\tau_b+w} Y_n^2   \right]- (\Exp[R_b   ])^2.
\end{align}
 
Moreover, by Wald's identity we have 
\begin{align} \label{wald22}
\begin{split}
 \Exp \left[ \sum_{n=w+1}^{\tau_b} Y_n^2    \right]  
 &= \Exp\left[ \sum_{n=w+1}^{\tau_b} 
\Exp\left[Y^2_n \, |\, \cG_{n-1}\right]  \right]\\
&\stackrel{\text{(a)}}{\leq} v \, \Exp[\tau_b] \\
 &\stackrel{\text{(b)}}\leq K \,  \left(b+ \Exp[R_b  ] +w \mu \right),
 \end{split}
\end{align}
where step (a) in \eqref{wald22} holds by  assumption \eqref{variance} of the theorem, step (b) by \eqref{wald2}, and 
\[
K\equiv  \frac{v}{ \mu^* (1 -q/w)}.
\]

Combining \eqref{wald21} and \eqref{wald22}, we have 
$$ (\Exp[R_b ])^2   - K (b+ \Exp[R_b   ] +w \mu)  \leq 0.$$

Solving this quadratic inequality, we conclude that
$$\Exp[R_b ] \leq \frac{1}{2} \left( K+\sqrt{K^2+4K(b+w\mu)}  \right) \leq K+\sqrt{K b} + \sqrt{K w\mu},
$$
where the second inequality follows by the subadditivity of the square-root. This implies  \eqref{show} with  $C=\max\{K, \sqrt{K}, \sqrt{K\mu} \}$ and completes the proof. \\

(iii) This follows directly by (ii).

  \end{proof}

\bibliographystyle{IEEEtran}
\bibliography{ref,biblio_new}

\begin{thebibliography}{10}
\providecommand{\url}[1]{#1}
\csname url@samestyle\endcsname
\providecommand{\newblock}{\relax}
\providecommand{\bibinfo}[2]{#2}
\providecommand{\BIBentrySTDinterwordspacing}{\spaceskip=0pt\relax}
\providecommand{\BIBentryALTinterwordstretchfactor}{4}
\providecommand{\BIBentryALTinterwordspacing}{\spaceskip=\fontdimen2\font plus
\BIBentryALTinterwordstretchfactor\fontdimen3\font minus
  \fontdimen4\font\relax}
\providecommand{\BIBforeignlanguage}[2]{{%
\expandafter\ifx\csname l@#1\endcsname\relax
\typeout{** WARNING: IEEEtran.bst: No hyphenation pattern has been}%
\typeout{** loaded for the language `#1'. Using the pattern for}%
\typeout{** the default language instead.}%
\else
\language=\csname l@#1\endcsname
\fi
#2}}
\providecommand{\BIBdecl}{\relax}
\BIBdecl

\bibitem{isit_2022_pub}
G.~Fellouris and V.~V. Veeravalli, ``Quickest change detection with controlled
  sensing,'' in \emph{2022 IEEE International Symposium on Information Theory
  (ISIT)}.\hskip 1em plus 0.5em minus 0.4em\relax IEEE, 2022, pp. 1921--1926.

\bibitem{poor-hadj-qcd-book-2009}
H.~V. Poor and O.~Hadjiliadis, \emph{Quickest detection}.\hskip 1em plus 0.5em
  minus 0.4em\relax Cambridge University Press, 2009.

\bibitem{tart-niki-bass-2014}
A.~G. Tartakovsky, I.~V. Nikiforov, and M.~Basseville, \emph{Sequential
  Analysis: {Hypothesis} Testing and Change-Point Detection}, ser.
  Statistics.\hskip 1em plus 0.5em minus 0.4em\relax CRC Press, 2014.

\bibitem{veer-bane-e-ref-2013}
V.~V. Veeravalli and T.~Banerjee, ``Quickest change detection,'' in
  \emph{Academic press library in signal processing: Array and statistical
  signal processing}.\hskip 1em plus 0.5em minus 0.4em\relax Academic Press,
  2013.

\bibitem{xie-zou-xie-vvv-jsait-qcd-survey}
L.~Xie, S.~Zou, Y.~Xie, and V.~V. Veeravalli, ``Sequential (quickest) change
  detection: Classical results and new directions,'' \emph{IEEE Journal on
  Selected Areas in Information Theory}, vol.~2, no.~2, pp. 494--514, 2021.

\bibitem{zhang-mei2020banditQCD_journal}
\BIBentryALTinterwordspacing
W.~Zhang and Y.~Mei, ``Bandit change-point detection for real-time monitoring
  high-dimensional data under sampling control,'' \emph{Technometrics},
  vol.~65, no.~1, pp. 33--43, 2023. [Online]. Available:
  \url{https://doi.org/10.1080/00401706.2022.2054861}
\BIBentrySTDinterwordspacing

\bibitem{banditQCD}
A.~Gopalan, B.~Lakshminarayanan, and V.~Saligrama, ``Bandit quickest
  changepoint detection,'' \emph{Advances in Neural Information Processing
  Systems}, vol.~34, 2021.

\bibitem{krishnamurthy2016pomdp}
V.~Krishnamurthy, \emph{Partially observed Markov decision processes: From
  filtering to controlled sensing}.\hskip 1em plus 0.5em minus 0.4em\relax
  Cambridge university press, 2016.

\bibitem{cher-amstat-1959}
H.~Chernoff, ``Sequential design of experiments,'' \emph{Ann. Math. Statist.},
  vol.~30, pp. 755--770, 1959.

\bibitem{Bessler1960_I}
S.~A. Bessler, ``Theory and applications of the sequential design of
  experiments, k-actions and infinitely many experiments, {Part I: T}heory.''
  Department of Statistics, Stanford University, Technical Report~55, 1960.

\bibitem{Bessler1960_II}
------, ``Theory and applications of the sequential design of experiments,
  k-actions and infinitely many experiments, {Part II: A}pplications.''
  Department of Statistics, Stanford University, Technical Report~56, 1960.

\bibitem{albert1961}
A.~E. Albert, ``{The Sequential Design of Experiments for Infinitely Many
  States of Nature},'' \emph{The Annals of Mathematical Statistics}, vol.~32,
  no.~3, pp. 774 -- 799, 1961.

\bibitem{Kiefer_Sacks_1963}
J.~Kiefer and J.~Sacks, ``{Asymptotically Optimum Sequential Inference and
  Design},'' \emph{The Annals of Mathematical Statistics}, vol.~34, no.~3, pp.
  705 -- 750, 1963.

\bibitem{Lalley_Lorden_1986}
S.~P. Lalley and G.~Lorden, ``{A Control Problem Arising in the Sequential
  Design of Experiments},'' \emph{The Annals of Probability}, vol.~14, no.~1,
  pp. 136 -- 172, 1986.

\bibitem{Keener_1984}
R.~Keener, ``{Second Order Efficiency in the Sequential Design of
  Experiments},'' \emph{The Annals of Statistics}, vol.~12, no.~2, pp. 510 --
  532, 1984.

\bibitem{niti-atia-veer-ieeetac-2013}
S.~Nitinawarat, G.~Atia, and V.~Veeravalli, ``Controlled sensing for
  multihypothesis testing,'' \emph{IEEE Trans. Aut. Contr.}, vol.~58, pp.
  2451--2464, 2013.

\bibitem{naghshvar2013active}
M.~Naghshvar and T.~Javidi, ``Active sequential hypothesis testing,'' \emph{The
  Annals of Statistics}, vol.~41, no.~6, pp. 2703--2738, 2013.

\bibitem{niti-veer-sqa-2015}
S.~Nitinawarat and V.~Veeravalli, ``Controlled sensing for sequential
  multihypothesis testing with controlled markovian observations and
  non-uniform control cost,'' \emph{Sequential Analysis}, vol.~34, no.~1, pp.
  1--24, Feb. 2015.

\bibitem{deshmukh2021sequential}
A.~Deshmukh, V.~V. Veeravalli, and S.~Bhashyam, ``Sequential controlled sensing
  for composite multihypothesis testing,'' \emph{Sequential Analysis}, pp.
  1--38, 2021.

\bibitem{Cohen2015active}
K.~Cohen and Q.~Zhao, ``Asymptotically optimal anomaly detection via sequential
  testing,'' \emph{IEEE Transactions on Signal Processing}, vol.~63, no.~11,
  pp. 2929--2941, 2015.

\bibitem{Cohen2019nonlinearcost}
A.~Gurevich, K.~Cohen, and Q.~Zhao, ``Sequential anomaly detection under a
  nonlinear system cost,'' \emph{IEEE Transactions on Signal Processing},
  vol.~67, no.~14, pp. 3689--3703, 2019.

\bibitem{Cohen2020composite}
B.~Hemo, T.~Gafni, K.~Cohen, and Q.~Zhao, ``Searching for anomalies over
  composite hypotheses,'' \emph{IEEE Transactions on Signal Processing},
  vol.~68, pp. 1181--1196, 2020.

\bibitem{oddball_2018}
N.~K. Vaidhiyan and R.~Sundaresan, ``Learning to detect an oddball target,''
  \emph{IEEE Transactions on Information Theory}, vol.~64, no.~2, pp. 831--852,
  2018.

\bibitem{Tsopela_2019}
A.~Tsopelakos, G.~Fellouris, and V.~V. Veeravalli, ``Sequential anomaly
  detection with observation control,'' in \emph{2019 IEEE International
  Symposium on Information Theory (ISIT)}, 2019, pp. 2389--2393.

\bibitem{Tsopela_2020}
A.~Tsopelakos and G.~Fellouris, ``Sequential anomaly detection with observation
  control under a generalized error metric,'' in \emph{2020 IEEE International
  Symposium on Information Theory (ISIT)}, 2020, pp. 1165--1170.

\bibitem{tsopelakos2021sequential}
------, ``Sequential anomaly detection with sampling constraints,'' \emph{To
  appear in IEEE Transactions on Information Theory}, 2022.

\bibitem{Draga96}
V.~Dragalin, ``A simple and effective scanning rule for a multi-channel
  system,'' \emph{Metrika}, vol.~43, no.~1, pp. 165--182, 1996.

\bibitem{Mei_Moust_2021}
Q.~Xu, Y.~Mei, and G.~V. Moustakides, ``Optimum multi-stream sequential
  change-point detection with sampling control,'' \emph{IEEE Transactions on
  Information Theory}, vol.~67, no.~11, pp. 7627--7636, 2021.

\bibitem{Anamitra2021}
A.~Chaudhuri, G.~Fellouris, and A.~Tajer, ``Sequential change detection of a
  correlation structure under a sampling constraint,'' in \emph{2021 IEEE
  International Symposium on Information Theory (ISIT)}, 2021, pp. 605--610.

\bibitem{XuMeiPostUncert_SeqAn}
Q.~Xu and Y.~Mei, ``Asymptotic optimality theory for active quickest detection
  with unknown postchange parameters,'' \emph{Sequential Analysis}, vol.~42,
  no.~2, pp. 150--181, 2023.

\bibitem{lai1998}
T.~L. Lai, ``Information bounds and quick detection of parameter changes in
  stochastic systems,'' \emph{IEEE Transactions on Information Theory},
  vol.~44, no.~7, pp. 2917--2929, November 1998.

\bibitem{xie-mous-xie-IT-2023}
L.~Xie, G.~V. Moustakides, and Y.~Xie, ``Window-limited cusum for sequential
  change detection,'' \emph{IEEE Transactions on Information Theory}, vol.~69,
  no.~9, pp. 5990--6005, 2023.

\bibitem{lorden1971}
G.~Lorden, ``Procedures for reacting to a change in distribution,'' \emph{The
  Annals of Mathematical Statistics}, vol.~42, no.~6, pp. 1897--1908, Dec.
  1971.

\bibitem{moulin-veeravalli-2018}
P.~Moulin and V.~V. Veeravalli, \emph{Statistical Inference for Engineers and
  Data Scientists}.\hskip 1em plus 0.5em minus 0.4em\relax Cambridge, UK:
  Cambridge University Press, 2019.

\bibitem{tartakovsky_book_2014}
A.~Tartakovsky, I.~Nikiforov, and M.~Basseville, \emph{Sequential analysis:
  Hypothesis testing and changepoint detection}.\hskip 1em plus 0.5em minus
  0.4em\relax CRC Press, 2014.

\bibitem{Lorden_1970}
\BIBentryALTinterwordspacing
G.~Lorden, ``{On Excess Over the Boundary},'' \emph{The Annals of Mathematical
  Statistics}, vol.~41, no.~2, pp. 520 -- 527, 1970. [Online]. Available:
  \url{https://doi.org/10.1214/aoms/1177697092}
\BIBentrySTDinterwordspacing

\end{thebibliography}

%\appendix

\end{document}